\documentclass[runningheads]{llncs}
\usepackage{graphicx}
\usepackage{amsmath}
\usepackage{cleveref}
\usepackage{url}

\newcommand{\Substr}{\mathit{Substr}}

\begin{document}
\title{On repetitiveness measures of Thue-Morse words}
\author{Kanaru Kutsukake\inst{1}\and
  Takuya Matsumoto\inst{1}\and
  Yuto~Nakashima\inst{1}\orcidID{0000-0001-6269-9353}\and
  Shunsuke~Inenaga\inst{1,2}\orcidID{0000-0002-1833-010X}\and
  Hideo Bannai\inst{3}\orcidID{0000-0002-6856-5185}\and
  Masayuki Takeda\inst{1}\orcidID{0000-0002-6138-1607}
}
\authorrunning{K. Kutsukake et al.}
\institute{Department of Informatics, Kyushu University, Fukuoka, Japan
  \and
  {PRESTO, Japan Science and Technology Agency, Kawaguchi, Japan}
  \and
  {M\&D Data Science Center, Tokyo Medical and Dental University, Tokyo, Japan}
  \email{\{kutsukake.kanaru,matsumoto.takuya,yuto.nakashima,\\inenaga,takeda\}@inf.kyushu-u.ac.jp}\\
  \email{hdbn.dsc@tmd.ac.jp}
}
\maketitle              \begin{abstract}
  We show that the size $\gamma(t_n)$ of the smallest string attractor of the
  $n$-th Thue-Morse word $t_n$ is 4 for any $n\geq 4$, disproving the conjecture by
  Mantaci et al. [ICTCS 2019] that it is $n$.
  We also show that 
  $\delta(t_n) = \frac{10}{3+2^{4-n}}$ for $n \geq 3$,
  where $\delta(w)$ is
  the maximum over all $k = 1,\ldots,|w|$,
  the number of distinct substrings of length $k$
  in $w$ divided by $k$,
  which is a measure of repetitiveness
  recently studied by Kociumaka et al. [LATIN 2020].  
  Furthermore, we show that the number $z(t_n)$ of factors
  in the self-referencing Lempel-Ziv factorization
  of $t_n$ is exactly $2n$.
\end{abstract}
\section{Introduction}
Measures which indicate the repetitiveness in a string
is a hot and important topic
in the field of string compression.
For example, given string $w$,
the size $g(w)$ of the smallest grammar that derives solely $w$~\cite{DBLP:journals/tit/CharikarLLPPSS05},
the number $z(w)$ of factors in the Lempel-Ziv factorization~\cite{DBLP:journals/tit/LempelZ76},
the number $r(w)$ of runs in the Burrows-Wheeler transform~\cite{Burrows94ablock-sorting} (RLBWT),
and the size $b(w)$ of the smallest bidirectional scheme (or macro schemes)~\cite{DBLP:journals/jacm/StorerS82}.
Recently, Kempa and Prezza
proposed the notion of {\em string attractor}~\cite{DBLP:conf/stoc/KempaP18},
and showed that the size $\gamma(w)$ of the smallest string attractor of
$w$ is a lower bound on
the size of the compressed representation for
these dictionary compression schemes.
While $z(w)$ and $r(w)$ are known to be
computable in linear time,
it is NP-hard to compute $g(w),b(w),\gamma(w)$~\cite{DBLP:journals/corr/abs-1811-12779,DBLP:journals/jacm/StorerS82,DBLP:conf/stoc/KempaP18}.

To further understand these measures,
Mantaci et al.~\cite{DBLP:conf/ictcs/MantaciRRRS19}
studied the size of the smallest string attractor in several well-known families of strings.
In particular,
they showed a size-2 string attractor for standard Sturmian words which is the smallest possible.
They further showed a string attractor of size $n$
for the $n$-th Thue-Morse word $t_n$, and conjectured it to be the smallest.

In this paper, we continue this line of work, and
investigate the exact values of various repetitive measures of the
$n$-th Thue-Morse word $t_n$.
More specifically,
we show that the size $\gamma(t_n)$ of the smallest string attractor of $t_n$
is $4$ for $n \geq 4$, disproving Mantaci et al.'s conjecture. Furthermore, we
give the exact value $\delta(t_n) = \frac{10}{3 + 2^{4-n}}$ for $n \geq 3$,
of the repetitiveness measure recently studied by Kociumaka et al.~\cite{KNPlatin20},
and the size $z(t_n)=2n$ of the self-referencing LZ77 factorization.

We note that for any standard Sturmian word $s$,
$z(s) =\Theta(\log|s|)$~\cite{DBLP:conf/mfcs/BerstelS06},
while the size $r(s)$ of the RLBWT is
always constant~\cite{DBLP:journals/ipl/MantaciRS03}.
On the other hand, $z(t_n)$ and $r(t_n)$ are both $\Theta(n)$, i.e.,
logarithmic in the length $|t_n|$
(the former due to~\cite{DBLP:conf/mfcs/BerstelS06} as well as
this work,
and the latter due to~\cite{DBLP:conf/iwoca/BrlekFMPR19}).
This shows that Thue-Morse words are an example where
the size of smallest string attractor is {\em not} a tight lower bound
on the size of the smallest of the
known efficiently computable dictionary compressed representations,
namely, $\min\{z(w), r(w)\}$.
We also conjecture that $b(t_n) = \Theta(n)$, which would seem to imply that
the size of the smallest string attractor
is not a tight lower bound for {\em all} currently known dictionary compression schemes.

Let $\ell(w)$ denote the size of the Lyndon factorization~\cite{ChenFoxLyndon58} of $w$.
It is known that
for any $w$, $\ell(w) = O(g(w))$~\cite{DBLP:conf/spire/INIBT13}
and $\ell(w) = O(z(w))$~\cite{DBLP:conf/stacs/KarkkainenKNPS17,DBLP:conf/cpm/UrabeNIBT19},
although it can be much smaller.
Interestingly, it is also known that
$\ell(t_n) = \Theta(n)$ (Theorem 3.1, Remark 3.8 of~\cite{DBLP:journals/dmtcs/IdoM97}).
Thus, if $b(t_n) = \Theta(n)$,
then $\ell(t_n)$ would be an asymptotically tight lower bound
for the smallest size of known dictionary compression schemes
for $t_n$, while $\gamma(t_n)$ is not.

Table~\ref{tab:summary} summarizes what we know so far.

\begin{table}
  \caption{Repetitiveness measures for the $n$-th Thue-Morse word $t_n$.}\label{tab:summary}
  \centerline{
    \begin{tabular}{|c|p{5cm}|c|c|c|}\hline
      measure       & description                                             & value                                                 & reference                                   \\\hline
$z(t_n)$      & Size of Lempel-Ziv factorization with self-reference    & $2n$                                                  & \cite{DBLP:conf/mfcs/BerstelS06}, this work \\\hline
      $r(t_n)$      & Number of same-character runs in BWT                    & $2n$                                                  & \cite{DBLP:conf/iwoca/BrlekFMPR19}          \\\hline
      $\ell(t_n)$   & Size of Lyndon factorization                            & $\displaystyle{\Big\lfloor\frac{3n-2}{2}\Big\rfloor}$ & \cite{DBLP:journals/dmtcs/IdoM97}           \\\hline
      $b(t_n)$      & Size of smallest bidirectional scheme                   & open                                                  & N/A                                         \\\hline
      $\gamma(t_n)$ & Size of smallest string attractor                       & $4$ ($n \geq 4$)                                      & this work                                   \\\hline
      $\delta(t_n)$ & maximum of subword complexity divided by subword length & $\displaystyle{\frac{10}{3 + 2^{4-n}}}$ ($n \geq 3$)  & this work                                   \\\hline
    \end{tabular}
  }
\end{table}
 \section{Preliminaries}
Let $\Sigma$ denote a set of symbols called the alphabet.
An element of $\Sigma^*$ is called a string.
For any $k \geq 0$, let $\Sigma^k$ denote the set of strings of length exactly $k$.
For any string $w$,
the length of $w$ is denoted by $|w|$.
For any $1 \leq i \leq |w|$,
let $w[i]$ denote the $i$th symbol of $w$, and
for any $1 \leq i \leq j \leq |w|$,
let $w[i..j] = w[i] w[i+1]\cdots w[j]$.

If $w = xyz$ for strings $x,y,z\in\Sigma^*$, then
$x,y,z$ are respectively called a prefix, substring, suffix of $w$.
We denote by $\Substr(w)$, the set of substrings of $w$.

In this paper, we will only consider the binary alphabet $\Sigma = \{\mathtt{a},\mathtt{b}\}$.
For any string $w\in\Sigma^*$,
let $\overline{w}$ denote the string obtained from $w$ by
changing all occurrences of $\mathtt{a}$ (resp. $\mathtt{b}$)
to $\mathtt{b}$ (resp. $\mathtt{a}$).
\begin{definition}[Thue-Morse Words~\cite{prouhet1851,thue1906,morse1921}]
  The $n$-th Thue-Morse word $t_n$ is a string over a binary alphabet
  $\{\mathtt{a},\mathtt{b}\}$
  defined recursively as follows:
  $t_0 = \mathtt{a}$, and for any $n>0$,
  $t_n = t_{n-1}\overline{t_{n-1}}$.
\end{definition}

It is a simple observation that $|t_n| = 2^n$ for any $n \geq 0$.

Below, we define the repetitiveness measures used in this paper:
\begin{description}
\item[String attractors~\cite{DBLP:conf/stoc/KempaP18}]
        For any string $w$, a set $\Gamma$ of positions in $w$ is a string attractor of $w$,
        if, for any substring $x$ of $w$, there is an occurrence of $x$ in $w$ that
        contains a position in $\Gamma$.
        For any string $w$, we will denote the size of a smallest string
        attractor of $w$ as $\gamma(w)$.
  \item [$\delta$~\cite{DBLP:journals/algorithmica/RaskhodnikovaRRS13,KNPlatin20}]
  \item For any string $w$,
        \[
          \delta(w) = \max_{k=1,..,|w|}\left(|\Sigma^k \cap \Substr(w)|/k\right).
        \]
  \item[LZ factorization~\cite{DBLP:journals/tit/LempelZ76}]
        For any string $w$, the LZ factorization of $w$ is the sequence
        $f_1, \ldots, f_z$ of non-empty strings such that
        $w = f_1\cdots f_z$,
        and for any $1 \leq i \leq z$,
        $f_i$ is the longest prefix of $f_i\cdots f_z$ which has at least two occurrences in $f_1\cdots f_i$, or, $|f_i| = 1$ otherwise.
        We denote the size of the LZ factorization of string $w$ as $z(w)$.

\end{description}

It is known that $\delta(w) \leq \gamma(w) \leq z(w), r(w)$ for any $w$~\cite{DBLP:journals/corr/abs-1811-12779,DBLP:conf/stoc/KempaP18}.

 \section{Repetitive Measures of Thue-Morse Words}

\subsection{$\gamma(t_n)$}

Mantaci et al.~\cite{DBLP:conf/ictcs/MantaciRRRS19} showed
the following explicit string attractor of size $n$ for the $n$-th Thue-Morse word.
\begin{theorem}[Theorem~8 of~\cite{DBLP:conf/ictcs/MantaciRRRS19}]
  \label{thm:mantaci2019}
  A string attractor of the $n$-th Thue Morse word, with $n\geq 3$ is
  \[ \left\{2^{n-1}+1\right\}\cup \{3\cdot 2^{i-2}\mid i = 2,\ldots,n\}.\]
\end{theorem}

To prove our new upperbound of $4$
for the smallest string attractor of $t_n$
for $n\geq 4$, we first show the following lemma.
\begin{lemma}\label{lem:kushi}
  Let
  \[
    N_n = \{ t_{n-1}\overline{t_{n-1}}\}\cup
\left(\bigcup_{k=0}^{n-2}\{ t_k\overline{t_k}, \overline{t_k}t_k\}\right).
  \]
  Then, for any substring $w\in\Substr(t_n)$ and $n\geq 2$,
  there exists $s \in N_n$
  such that the occurrence of $w$ in $s$
  contains the center of $s$ (i.e., position $|s|/2$).
\end{lemma}
\begin{proof}
  Consider the recursively defined perfect binary tree
  with $t_n$ as the root,
  with $t_{n-1}$ and $\overline{t_{n-1}}$
  respectively as its left and right children (See Fig.~\ref{fig:attractor}).
  The leaves consist of either $t_0$ or $\overline{t_0}$,
  each corresponding to a position of $t_n$.
  If $|w| = 1$, then, we can choose
  $t_1 = t_{0}\overline{t_0} = \mathtt{ab}$ for $\mathtt{a}$ and
  $t_2 = t_{1}\overline{t_{1}}=\mathtt{abba}$ for $\mathtt{b}$.
  For any substring $w = t_n[i..j]$ of length at least 2,
  consider the lowest common ancestor
  of leaves corresponding to $t_n[i]$ and $t_n[j]$.
  Each node of the tree is
  $t_n = t_{n-1}\overline{t_{n-1}}$ if it is the root,
  or otherwise, either
  $t_{k+1}=t_k\overline{t_k}$
  or $\overline{t_{k+1}} = \overline{t_k}t_k$ for some $0 \leq k \leq n-2$.
  Since
  $w$ is a substring that starts in the left child and ends in the right child
  of the lowest common ancestor,
  the occurrence of $w$ must contain the center,
  and the lemma holds.
  \qed
\end{proof}

\begin{figure}
  \centerline{
    \includegraphics[width=0.7\textwidth]{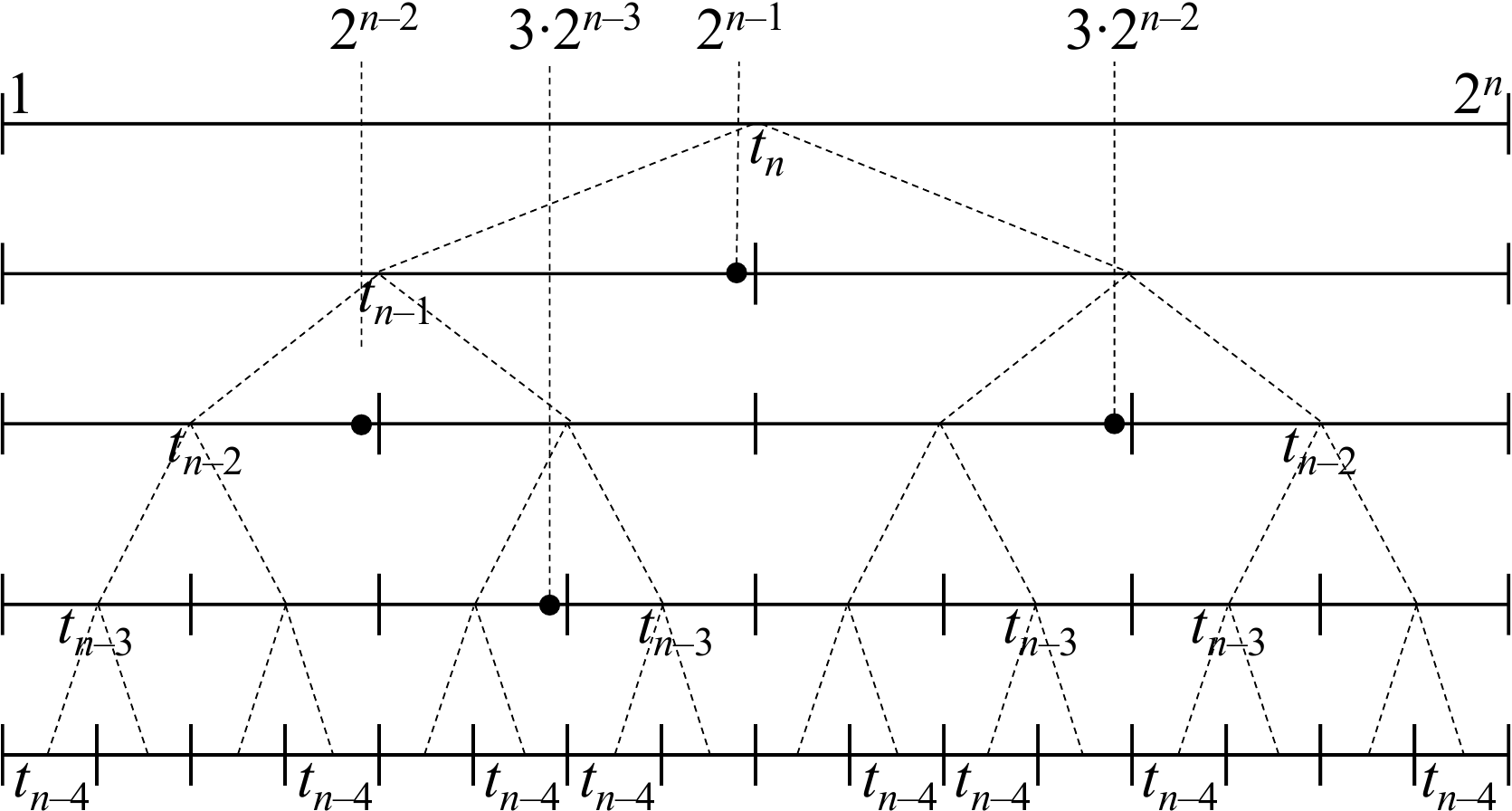}
  }
  \caption{A representation of $t_n$ as a perfect binary tree (shown
    to depth 4) introduced in the proof of~\Cref{lem:kushi}.
    For each level where segments are labeled with $t_k$,
    non-labeled segments represent $\overline{t_k}$.
    The black circles depict the four positions in $K_n$ defined in
    \Cref{thm:gamma_tn},
    at the node at which the center of the parent
    coincides with the position.
  }\label{fig:attractor}
\end{figure}

\begin{theorem}\label{thm:gamma_tn}
  For any $n \geq 4$, the set
  \[
    K_n  = \left\{
    2^{n-2}, 3\cdot 2^{n-3}, 2^{n-1},3\cdot 2^{n-2}
\right\}
  \]
  is a string attractor of $t_n$.
\end{theorem}
\begin{proof}

  Let $w$
  be an arbitrary substring of $t_n$. From Lemma~\ref{lem:kushi}, it suffices to show that
  any element in $N_n$ has an occurrence in
  $t_n$ whose center coincides with a position in $K_n$.
$t_{n-1}\overline{t_{n-1}}$,
  $t_{n-2}\overline{t_{n-2}}$,
  $\overline{t_{n-2}}t_{n-2}$, and
  $\overline{t_{n-3}}t_{n-3}$
  each have an occurrence whose center coincides respectively with
  position $2^{n-1}$, $2^{n-2}$, $3\cdot 2^{n-2}$, and $3\cdot 2^{n-3}$
  which are all elements of $K_n$ (see Fig.~\ref{fig:attractor}).
  Furthermore,
  there is an occurrence of $t_{n-3}\overline{t_{n-3}}$ whose center coincides with
  that of $t_{n-1}\overline{t_{n-1}}$, and thus with an element of $K_n$.
  More generally, for any $2\leq k \leq n-2$,
  each occurrence of $t_k\overline{t_{k}}$ implies
  an occurrence of $t_{k-2}\overline{t_{k-2}}$
  whose centers coincide.
  This is because
  \begin{eqnarray*}
    t_{k}\overline{t_{k}}
    &=& t_{k-1}\overline{t_{k-1}}\overline{t_{k-1}}t_{k-1}\\
    &=& t_{k-1}\overline{t_{k-2}}t_{k-2}\overline{t_{k-2}}t_{k-2} t_{k-1}.
  \end{eqnarray*}
  The same argument holds for $\overline{t_{k-2}}t_{k-2}$ by considering
  $\overline{t_{k}}t_{k}$.
  The theorem follows from a simple induction.
\qed
\end{proof}

\begin{theorem}
  $\gamma(t_n) = 4$ for any $n \geq 4$.
\end{theorem}
\begin{proof}
  \Cref{thm:gamma_tn} implies $\gamma(t_n) \leq 4$.
  From \Cref{thm:delta_tn} shown in the next subsection,
  we have $\delta(t_n) > 3$ for $n \geq 6$.
  Since $\gamma(t_n)$ is an integer which cannot
  be smaller than $\delta(t_n)$,
  it follows that $\gamma(t_n) \geq 4$ for $n \geq 6$.
  For $n = 4,5$,
  it can be shown by exhaustive search that
  there is no string attractor of size~$3$.
  \qed
\end{proof}

\subsection{$\delta(t_n)$}
Brlek~\cite{BRLEK198983} investigated the number of
distinct substrings of length $m$ in $t_n$, and gave an exact formula.
Below is a summary of his result which will be a key to computing $\delta(t_n)$.

\begin{lemma}[Proposition 4.2, Corollary 4.2.1, Proposition 4.4 of~\cite{BRLEK198983}]
  The number $P_n(m)$ of distinct substrings of length $m\geq 3$ in $t_n~(n \geq 3)$
  is:
  \[
    P_n(m) =
    \begin{cases}
      2^n-m+1           & 2^{n-2}+1\leq m\leq 2^n                     \\
      6\cdot 2^{q-1}+4p & 3 \leq m \leq 2^{n-2}, 0 < p \leq 2^{q-1}   \\
      8\cdot 2^{q-1}+2p & 3 \leq m \leq 2^{n-2}, 2^{q-1} < p \leq 2^q
    \end{cases}
  \]
  where $p,q$ are values uniquely determined by
  $m = 2^q+p+1$ and $0 < p \leq 2^q$.
\end{lemma}

\begin{theorem}\label{thm:delta_tn}
  \[
    \delta(t_n) =
    \begin{cases}
      1                    & n = 0    \\
      2                    & n = 1,2  \\
      \frac{10}{3+2^{4-n}} & n \geq 3 \\
    \end{cases}
  \]
\end{theorem}
\begin{proof}
  We only consider $n\geq 3$ below.
  The number of distinct substrings of length $1$ and $2$ in $t_n$,
  are respectively $2$ and $4$.
  For $2^{n-2}+1 \leq m \leq 2^n$,
  \[
    \max_{2^{n-2}+1 \leq m \leq 2^n} \frac{P_n(m)}{m} =
    \max_{2^{n-2}+1 \leq m \leq 2^n} \left\{\frac{2^n+1}{m} - 1 \right\}= \frac{2^n+1}{2^{n-2}+1} - 1 = \frac{3}{1+2^{2-n}}.
  \]
  For $3 \leq m \leq 2^{n-2}$ and fixed $q$,
  it is easy to verify that $P_n(m)/m$ is increasing when
  $0 < p \leq 2^{q-1}$, and non-increasing when $2^{q-1}<p\leq 2^q$,
  because
  \[
    \left(\frac{6\cdot 2^{q-1}+4p}{2^q+p+1}\right)'
    = \frac{4(2^q+p+1) - (6\cdot2^{q-1}+4p)}{(2^q + p + 1)^2} = \frac{2^q+4}{(2^q+p+1)^2} > 0
  \]
  and
  \[
    \left(\frac{8\cdot 2^{q-1}+2p}{2^q+p+1}\right)'
    = \frac{2(2^q+p+1) - (8\cdot2^{q-1}+2p)}{(2^q+p+1)^2}
    = \frac{(2-4\cdot2^{q-1})}{(2^q+p+1)^2} \leq 0.
  \]
  Also note that $6\cdot 2^{q-1}+4p = 8\cdot 2^{q-1}+2p$ when $p = 2^{q-1}$.
  Therefore, for a fixed $q$, the maximum value of
  $\frac{P_n(m)}{m}$ is obtained when $p = 2^{q-1}$,
  i.e.,
  $\frac{6\cdot 2^{q-1}+4\cdot 2^{q-1}}{2^q + 2^{q-1}+1} = \frac{10\cdot2^{q-1}}{3\cdot2^{q-1}+1} = \frac{10}{3+2^{1-q}}$.
  Since this is increasing in $q$, we have that
  $\max_{3 \leq m \leq 2^{n-2}}\frac{P_n(m)}{m}$ is obtained by
  choosing the largest possible $q = n-3$
  (where
  $p = 2^{q-1} = 2^{n-4}$, and thus $m = 2^{n-3} + 2^{n-4}+ 1 = 3\cdot2^{n-4}+1 \leq 2^{n-2}$),
  which gives us the final result $\delta(t_n) = \max\{\frac{2}{1},\frac{4}{2},\frac{10}{3+2^{4-n}}, \frac{3}{1+2^{2-n}}\} = \frac{10}{3+2^{4-n}}$.
  \qed
\end{proof}

\subsection{LZ77}
We consider the size $z(t_n)$ of the LZ factorization.
Although Berstel and Savelli~\cite{DBLP:conf/mfcs/BerstelS06}
have given a complete characterization of the LZ factorization for the infinite Thue-Morse word,
we show an alternate proof in terms of the $n$-th Thue-Morse word.
Below is an important lemma, again by Brlek, we will use.
\begin{lemma}[Corollary 4.1.1 of~\cite{BRLEK198983}]
  \label{lem:uniqueSubstrings}
  The word $t_n$ has one and only one occurrence
  of every factor $w$ such that $|w| \geq 2^{n-2}+1$.
\end{lemma}

\begin{theorem}\label{thm:lz_tn}
  For any $n \geq 1$, $z(t_n) = 2n$.
\end{theorem}
\begin{proof}
  Clearly, $z(t_1) = 2$.
  Since
  $t_k = t_{k-1}\overline{t_{k-1}} = t_{k-2}\overline{t_{k-2}}\overline{t_{k-2}}t_{k-2}$, it is easy to see that $z(t_k) \leq z(t_{k-1}) + 2$,
  because $\overline{t_{k-2}}$ and $t_{k-2}$ respectively
  have earlier occurrences in $t_k$.
  Thus, $z(t_n) \leq 2n$.
  On the other hand,
  \Cref{lem:uniqueSubstrings} implies that
  the substring $t_k[2^{k-1}..3\cdot2^{k-2}]$
  of length $2^{k-2}+1$ cannot be a single LZ factor,
  implying that position $2^{k-1} (= |t_{k-1}|)$
  and position $3\cdot{2^{k-2}} (> |t_{k-1}|)$ belong to different factors.
  Similarly, the substring
  $t[3\cdot 2^{k-2}..2^{k}]$
  of length $2^{k-2}+1$
  cannot be a single LZ factor, implying that
  position $3\cdot{2^{k-2}}$ and position $2^{k}$ belong to different factors.
  Thus,  $z(t_{k}) \geq z(t_{k-1}) + 2$, implying $z(t_n) \geq 2n$.
  \qed
\end{proof}

 \section*{Acknowledgments}
This work was supported by JSPS KAKENHI Grant Numbers
JP18K18002 (YN),
JP17H01697 (SI),
JP16H02783, JP20H04141 (HB), JP18H04098 (MT),
and JST PRESTO Grant Number JPMJPR1922 (SI). \clearpage
\bibliographystyle{splncs04}
\bibliography{ref}
\end{document}